\documentclass[11pt,a4paper]{article}
\usepackage{fullpage}
\usepackage{booktabs}
\usepackage{xcolor}
\usepackage{multirow}
\usepackage[utf8]{inputenc}
\usepackage{amsthm,amssymb}
\usepackage[leqno]{amsmath}
\usepackage{tcolorbox}
\usepackage{thmtools}
\usepackage{mathtools}
\usepackage{subcaption,thm-restate}
\usepackage[mathcal,mathscr]{eucal}
\usepackage{epsfig,graphicx,graphics,color}
\usepackage{caption}
\usepackage{enumerate}
\usepackage{enumitem}
\usepackage[sort]{cite}
\usepackage{titling}
\usepackage{hyperref}
\usepackage{rotating} 

\hypersetup{
    colorlinks=true,       
    linkcolor=blue,        
    citecolor=red,         
    filecolor=magenta,     
    urlcolor=cyan,         
    linktocpage=true
}
\usepackage{thm-restate}
\usepackage{cleveref}
\usepackage{float}
\usepackage{algpseudocode}
\usepackage{rotating}
\usepackage{makecell}
\usepackage[ruled,noline, linesnumbered]{algorithm2e}

\theoremstyle{plain}
\newtheorem{theorem}{Theorem}
\newtheorem{lemma}[theorem]{Lemma}

\theoremstyle{definition}

\DeclareMathOperator*{\argmin}{arg\,min}
\DeclareMathOperator{\lcm}{lcm}

\DeclarePairedDelimiter\floor{\lfloor}{\rfloor}

\def\final{0}  
\def\iflong{\iffalse}
\ifnum\final=0  
\newcommand{\anote}[1]{{\color{orange}[{\tiny \textbf{Alex:} \bf #1}]\marginpar{\color{orange}*}}}

\newcommand{\jnote}[1]{{\color{blue}[{\tiny \textbf{Julian:} \bf #1}]\marginpar{\color{blue}*}}}

\else 
\newcommand{\anote}[1]{}
\newcommand{\knote}[1]{}
\newcommand{\jnote}[1]{}
\newcommand{\lnote}[1]{}
\fi

\usepackage{authblk}

\title{Optimizing Periodic Operations for Efficient Inland Waterway Lock Management}

\date{}

\thanksmarkseries{alph}

\author{Julian Golak}
\affil{{\footnotesize Institute of Operations Management, University of Hamburg Business School, Hamburg, Germany. \\ \texttt{julian.golak@uni-hamburg.de}.}\footnote{Corresponding author}}
\author{Alexander Grigoriev}
\author{Freija van Lent}
\author{Tom van der Zanden}
\affil{{\footnotesize Department of Data Analytics and Digitalisation, Maastricht University, The Netherlands. \\ \texttt{a.grigoriev@maastrichtuniversity.nl, f.vanlent@maastrichtuniversity, t.vanderzanden@maastrichtuniversity.nl}.} }

\begin{document}
\maketitle

\begin{abstract}
    In inland waterways, the efficient management of water lock operations impacts the level of congestion and the resulting uncertainty in inland waterway transportation. 
    To achieve reliable and efficient traffic, schedules should be easy to understand and implement, reducing the likelihood of errors.
    The simplest schedules follow periodic patterns, reducing complexity and facilitating predictable management.
    Since vessels do not arrive in perfectly regular intervals, periodic schedules may lead to more wait time.
    The aim of this research is to estimate this cost by evaluating how effective these periodic schedules manage vessel traffic at water locks.

    The first objective is to estimate a periodic arrival pattern that closely matches a dataset of irregular vessel arrivals at a specific lock.
    We develop an algorithm that, given a fixed number of vessel streams, solves the problem in polynomial time.
    The solution then serves as input for the subsequent part, where we consider algorithms that compute operational schedules by formulating an optimisation problem with periodic arrival patterns as input, and the goal is to determine a periodic schedule that minimises the long-run average waiting time of vessels.    
     We present a polynomial-time algorithm for the two-stream case and a pseudo-polynomial-time algorithm for the general case, along with incremental polynomial-time approximation schemes.
     
    In our numerical experiments, use AIS data to construct a periodic arrival pattern closely matching the observed data. Our experiments demonstrate that when evaluated against actual data, intuitive and straightforward policies often outperform optimal policies specifically trained on the periodic arrival pattern.
\medskip


\end{abstract}
\section{Introduction}
    The European Green Deal, introduced by the European Commission \cite{europeancommission2020greendeal}, is a comprehensive plan with the ambitious goal of transforming Europe into the world's first climate-neutral continent by 2050. 
    A central focus of this plan is the promotion of sustainable transport modes as alternatives to traditional road transport. Inland waterways play a pivotal role in this context, offering several environmental advantages that align with the Green Deal's sustainability goals, including reduced greenhouse gas emissions, improved air quality, and decreased road traffic congestion. 
    Further, the Green Deal identifies the development of advanced algorithms, data-driven approaches, and intelligent decision support systems as crucial components for driving the transition towards a more sustainable, decarbonized transport system.
    To fully leverage the potential of inland waterways, efficient and integrated waterway management is essential. 
    A critical component of integrated waterway management is the operation of locks. 
    During peak times, the number of arriving vessels often exceeds the lock’s throughput capacity, leading to congestion \cite{golak2022optimizing}. 
    Strategically scheduling lock operations can reduce congestion and delays, enhancing the efficiency and attractiveness of inland waterway transportation.

    Operational schedules should be simple and intuitive, effectively reducing congestion while minimizing errors among operators. 
    Furthermore, simple schedules also enhance robustness and predictability for skippers. This simplicity and predictability can be achieved through periodic operational schedules, which organize operations at consistent, regular intervals. Implementing such periodic schedules represents an essential first step toward fully automated logistics, where lock operations are planned and managed in advance, enabling skippers to efficiently coordinate their voyages. The transition toward automation offers multiple advantages, including greater reliability and enhanced operational efficiency.
    
    In this work, we focus on algorithmic research on computing periodic operational schedules for the lock. That is, we consider the problem of finding a operational schedule for the lock, such that each operations repeats in regular time intervals.

    \paragraph{Related Literature.}
    Due to their practical importance, research on lock scheduling is vast.
    Previous research considered lock management mainly as a classical scheduling problem:
    Given a set of arriving vessels, the goal is to find a schedule of lock operations that minimizes objectives such as fuel consumption, emissions, or waiting time, while respecting the unique operational characteristics of the lock. 
    The most fundamental problem, addressed by Passchyn et al.~\cite{passchyn2016.3}, involved a single lock and a fleet of incoming vessels. They developed polynomial-time algorithms that  minimize the total waiting time.
    This was extended to include parallel lock chambers, allowing simultaneous vessel processing, by Passchyn, Briskorn, and Spieksma~\cite{passchyn2016.1}. 
    They provided a complexity analysis and devised polynomial-time solutions for specific cases. 
    Toward creating sustainable transportation, Passchyn, Briskorn, and Spieksma~\cite{passchyn2016.2} addressed the problem of choosing optimal vessel speeds for navigating a sequence of locks, aiming to minimize total fuel consumption and emissions.
    Golak, Defryn, and Grigoriev~\cite{golak2022optimizing} extended these results by developing near-optimal algorithms for determining vessel speeds in river networks with locks.
    The choice of speeds was further addressed in Buchem, Golak, and Grigoriev~\cite{buchem2022vessel}, where the authors designed algorithms to determine optimal vessel speeds while accounting for the uncertainty in lock operation durations.
    Defryn, Golak, Grigoriev, and Timmermans~\cite{defryn2021inland} developed a behavioral framework to determine vessel speeds, incorporating the assumption that skippers may coordinate their actions.
    Lastly, the strategic placement of vessels within lock chambers was further investigated by Verstichel et al.~\cite{verstichel2014.1, verstichel2014.2}. 

    Numerous case studies have been conducted for the lock scheduling problem, concentrating on specific lock sequences along crucial waterways worldwide. 
    Petersen and Taylor~\cite{petersen1988} consider the Welland Canal in North America, offering a heuristic that utilizes optimal dynamic programming models for scheduling individual locks to determine operating schedules for the lock sequence. On the Upper Mississippi River, Smith, Sweeney and Campbell~\cite{smith2009} present a simulation model to evaluate various heuristics for lock operations quality, an investigation extended by Smith et al.~\cite{smith2011} through the proposal of a mathematical program that addresses the lock scheduling problem with sequence-dependent setup and processing times. Meanwhile, Nauss~\cite{nauss2008} incorporates lock malfunctions and explores responses to minimize additional queue lengths on the same river segment. Additionally, Ting and Schonfeld~\cite{ting2001} investigate a model for the lock scheduling problem with multiple parallel chambers for this river layout. 
    Finally, L{\"u}bbecke, L{\"u}bbecke, and M{\"o}hring~\cite{lubbecke2019ship} examine the Kiel Canal, incorporating ship collision avoidance in an optimization model and presenting a heuristic for determining collision-free routing and scheduling for a fleet of ships.

    To our knowledge, the combination of periodic and lock scheduling has not been addressed in literature. However, related areas, such as the periodic maintenance problem, have received considerable attention.
    In Anily, Glass and Hassin~\cite{anily1998scheduling}, the authors discuss the free periodic maintenance problem, where the costs of operating a machine at a given time linearly depend on the time since its last maintenance. They give polynomial time algorithms for special cases and an approximation algorithm for the general problem. 
    The work was complemented by Bar-Noy et al.~\cite{bar2002minimizing} formulating a generalization of this problem and showing that the feasibility problem is NP-hard. 
    Another closely related area is the perfect periodic scheduling problem. The perfect periodic scheduling problem involves creating schedules where each task is executed at consistent, regular intervals, ensuring that all tasks are accommodated without conflicts. 
    Glass~\cite{glass1992feasibility} proposes algorithms for finding a feasible schedule for three products with three distinct periodicities. 
    Kim and Glass~\cite{kim2014perfect} extended these results by developing a polynomial-time test to determine the existence of a feasible schedule and a method for constructing such a schedule when one exists.
    In addition, Glass~\cite{glass1994feasibility} and Glass, Gupta and Potts~\cite{glass1994lot} extend the perfect periodic scheduling problem in the context of Economic Lot Scheduling. 
    Due to the relevant applications and the complexity of perfect periodic scheduling, several heuristics have been proposed, see Patil and  Garg~\cite{patil2006adaptive} and Brakerski, Nisgav and Patt-Shamir~\cite{brakerski2006general}. 
    Finally, Grigoriev, Van De Klundert, and Spieksma~\cite{grigoriev2006modeling} focused on the periodic maintenance problem, aiming to minimize the total service and operating costs for a set of machines. Their approach considered integral cost coefficients and time discretized into unit-length periods.
    The authors derive various formulations, including flow and set-partitioning formulations, revealing efficient linear programming relaxations and polynomial time solutions for larger problem instances, achieving the first exact solutions.
    
    \paragraph{Our Results.} 
    Our main contribution is to initialise algorithmic research on periodic lock scheduling. 
    In its full generality, the unique characteristics of lock operations result in scheduling problems of significant complexity. Factors such as accommodating vessels of varying sizes, managing their respective loading times, and accounting for uncertainties in these loading processes all contribute to this complexity.
    To fundamentally understand the system, it is essential to analyze a simplified version, examine its structure, and determine how this structure can be effectively utilized in algorithm design.
    Consequently, this article focuses on a simplified problem setting to clearly reveal and analyze the underlying structure of periodic lock scheduling problems.
    Most importantly, we restrict our analysis to a single lock, assuming constant loading times for all vessels and unlimited lock chamber capacity.
   
    To design algorithms for computing operational schedules, we require a characterization of the arrival process as a periodic pattern. 
    Thus, we formulate an optimisation problem, where the input data consists of vessel arrival data. The goal is to assign these vessel arrivals into distinct streams and, for each stream, determine a constant inter-arrival time between consecutive arrivals, such that a specified loss function is minimised.
    By exploiting the structural properties of this problem, we develop an exact algorithm that runs in polynomial time when the number of streams is constant.
    The solution of this problem, serves as input for the algorithms that compute operational schedules.
    We again formulate a combinatorial problem and find structural properties that are used in the subsequent algorithms.
    We design a polynomial-time algorithm while restricting to only two streams and a pseudo-polynomial time algorithm for the general case.
    Further, we show the existence of approximation schemes that run in incremental polynomial time.
    The algorithms have been implemented and tested using a real-world dataset derived from Automatic Identification System (AIS) data. 
    We investigate the computational cost of finding the periodic streams, its effectiveness in fitting the data, and the performance of the optimal operational schedules in comparison to alternative strategies that are commonly used in practice.
    
    \medskip The remainder of this paper is organized as follows. Section~\ref{sec:matching} addresses the problem of identifying periodic patterns within a given set of arriving vessels. Section~\ref{sec:scheduling} discusses the problem of computing optimal periodic schedules of lock operations. 
    Numerical study results are presented in Section~\ref{sec:numerical}. Finally, Section~\ref{sec:conclusions} summarizes the findings and outlines open problems for future research.
    
    \paragraph{Preliminaries.} We denote the sets of \emph{real numbers} and \emph{non-negative real numbers} by~\(\mathbb{R}\) and~\(\mathbb{R}_+\), respectively. Similarly, the sets of \emph{integers} and \emph{non-negative integers} are denoted by~\(\mathbb{Z}\) and~\(\mathbb{Z}_+\), respectively. The modulus operator \(\bmod\) is used for integers \(a\) and \(b\) to find the remainder of \(a\) divided by \(b\). The congruence relation \(\equiv\), used in the context of modulo arithmetic, denotes that two numbers are equivalent in terms of their remainders when divided by a specific modulus. In other words, \(a \equiv b \mod m\) means that \(a\) and \(b\) leave the same remainder when divided by \(m\). The \emph{greatest common divisor} of two integers \(a\) and \(b\), denoted \(\gcd(a, b)\), is the largest integer that divides both without leaving a remainder. The \emph{least common multiple} of a set of integers \(n_1, \ldots, n_k\) is denoted with \(\lcm(n_1, \ldots, n_k)\), is the smallest positive integer that is divisible by each of these integers without any remainder. For a real number \(x\), we use \(\lfloor x \rfloor\) to denote the greatest integer that is less than or equal to \(x\).

    \section{Identifying periodic patterns in vessel arrivals}\label{sec:matching}
    In this section, we discuss the problem of matching a set of irregularly arriving vessels with regular streams. In Section~\ref{sec:ProblemDef1}, we define the problem. In Section~\ref{sec:optSol1}, we present structural insights underlying optimal solutions, which are used in Section~\ref{sec:algorithm1} to develop an enumeration algorithm for computing the optimal solution. This optimal solution characterizes regular streams, and the algorithm runs in polynomial time for a fixed number of streams.

    \subsection{Problem definition} \label{sec:ProblemDef1} 
    Throughout this paper, we consider a waterway containing a single lock at which vessels are arriving in \emph{downstream} and \emph{upstream} directions. 
    Our aim is to find regular patterns that fit the irregular arriving vessels most closely. 
    We consider a time interval of length \(T\) and let \(S\) denote the set of vessels that arrive at the lock within this time interval.
    We define \(n\coloneqq |S|\).
    We denote the arrival time of vessel \(s\in S\) with \(t(s)\) with \(1 \leq t(s) \leq T\) and assume there are \(k\) regular streams with \(1 \leq k \leq n\). 
    Each regular stream \(i\) is characterized by its \emph{periodicity} \(\lambda_i\), that is, the inter-arrival time between every two arrivals, and its \emph{offset} \(\mu_i\), that is the time of the vessel that arrives first.
    We assume that \(0 < \lambda_i \leq T\) and \(1 \leq \mu_i \leq \lambda_i\) for all streams \(1 \leq i \leq k\).
    Define the decision variables that characterize the regular streams with the set of pairs \(\Gamma = \{(\mu_i, \lambda_i) \mid 1 \leq i \leq k\}\).

    Given a stream \(i\) with some \(\mu_i\) and \(\lambda_i\), let \(n_i \coloneqq T / \lambda_i\) and assume that \(n_i \in \mathbb{Z}_+\). Stream~\(i\) induces \(n_i\) time points, defined as \(A_i = \{ (\mu_i + \ell \lambda_i, i) \mid \ell = 0, \dots, n_i-1\}\).
    The joint set of time points is defined as \(A = \bigcup_{i} A_i\). 
    Furthermore, for \(a = (t,i) \in A\), we define \(\tau(a) = t\). 
    We refer to \(A\) as the set of \emph{matching points}. 
    A set of stream parameters \(\Gamma\) is considered \emph{feasible} if the number of matching points equals the number of vessels. Additionally, the periodicities must match the time interval.
    In other words, the set \(A\) must have a size of \(n\) and the frequency of each stream are integer values, that is \(T / \lambda_i\in \mathbb{Z}_+\).

    We aim to find the periodicities and offset for each stream and a matching between the induced matching points and arriving vessels, such that the deviation of matching points to the data is minimized. 
    We define a matching between the vessels and the matching points with \(X\colon S \to A\), where \(x(s) = a\) if \(s\) is matched to \(a\in A\). 
    A matching \(X\) is \emph{feasible} if it forms a bijection between vessels in \(S\) and matching points in \(A\). 
    Given a pair \((\Gamma, X)\), the cost of matching a vessel to a matching point is defined as the absolute deviation between the vessel’s actual arrival time and the time of the matching point.
    Specifically, for a vessel \(s \in S\), we define its cost by \(C_\Gamma(s) = |\tau(x(s)) - t(s)|\). We define the total cost of the matching \(X\) as \(C_\Gamma(X) = \sum_{s\in S} C_\Gamma(s)\).
    Our aim is to find a feasible \emph{solution} defined by the pair \((\Gamma, X)\), such that \(C_\Gamma(X)\) is minimized. 
    We refer to this as the \emph{periodic matching problem}.
     
    \subsection{Characterizing optimal solutions} \label{sec:optSol1}

    We propose a number of structural results that are used for proving the correctness of the algorithm that follows. Specifically, we show that there exists an optimal solution that has a very specific structure.

\begin{lemma}\label{lemma:interchangeOpt}
   Given an optimal solution \((\Gamma, X)\), for each \(s_1, s_2 \in S\) such that \(t(s_1) \leq t(s_2)\), it holds that \(\tau(x(s_1)) \leq \tau(x(s_2))\). In other words, the solution is order-preserving.
\end{lemma}

\begin{proof}
    Assume there is an optimal solution \((\Gamma, X)\) with \(\tau(x(s_1)) > \tau(x(s_2))\) for some \(s_1, s_2 \in S\) with \(t(s_1) \leq t(s_2)\). 
    We define a solution \((\Gamma', X')\), where \(\Gamma' = \Gamma\) and \(X'\) is defined by
    \[x'(s) = \begin{cases}
        x(s_1) & \text{if } s = s_2,\\
        x(s_2) & \text{if } s = s_1,\\
        x(s) & \text{otherwise}.
    \end{cases}\]
    Clearly, \((\Gamma', X')\) is a feasible solution. Furthermore,
    \begin{align*}
    C&_\Gamma(X) - C_\Gamma(X') 
    {}{}= 
    \sum_{s \in S} \left| \tau(x(s)) - t(s) \right| - \sum_{s \in S} \left| \tau(x'(s)) - t(s) \right| \\
    {}&{}= 
    \sum_{s \in S} \left| \tau(x(s)) - t(s) \right| - \left(\left| \tau(x'(s_1)) - t(s_1) \right| + \left| \tau(x'(s_2)) - t(s_2) \right| + \sum_{s \in S \setminus \{s_1, s_2\}} \left| \tau(x'(s)) - t(s) \right| \right) \\
    {}&{}= 
    \sum_{s \in S} \left| \tau(x(s)) - t(s) \right| - \left(\left| \tau(x(s_2)) - t(s_1) \right| + \left| \tau(x(s_1)) - t(s_2) \right| + \sum_{s \in S \setminus \{s_i, s_j\}} \left| \tau(x(s)) - t(s) \right| \right) \\
    {}&{}= 
    \left( \left| \tau(x(s_1)) - t(s_1) \right| + \left| \tau(x(s_2)) - t(s_2) \right| \right) - \left( \left| \tau(x(s_2)) - t(s_1) \right| + \left| \tau(x(s_1)) - t(s_2) \right| \right) \\
    &> 0.
\end{align*}
    using that \(\tau(x(s_1)) > \tau(x(s_2))\) and \(t(s_1) \leq t(s_2)\). This implies that \(C_\Gamma(X) > C_\Gamma(X')\), contradicting the optimality of \((\Gamma,X)\) and showing the claim.
\end{proof}

\begin{lemma}\label{lem:anchoring}
    Let \(S_i\) denote all vessels \(s\in S\) with \(x(s) \in A_i\).
    There exists an optimal solution \((\Gamma, X)\), such that for each stream \(i\), there exists a matching point \(x(s)\) with \(\tau(x(s)) = t(s)\) for some \(s\in S_i\). In other words, stream \(i\) is anchored at vessels \(s\).
\end{lemma}
\begin{proof}
    Consider an optimal solution \((\Gamma, X)\) with a stream \(i\) such that \(\tau(x(s)) \neq t(s)\) for all \(s \in S_i\). Define \(\Delta(s) \coloneqq t(s) - \tau(x(s))\) for \(s\in S_i\) and define \(\mathcal{D}^+ \coloneqq |\{s\in S_i \mid \Delta(s) > 0\}|\) and \(\mathcal{D}^- \coloneqq |\{s\in S_i \mid \Delta(s) < 0\}|\).

    We claim that is must hold that \(\mathcal{D}^- = \mathcal{D}^+\). 
    Assume \(\mathcal{D}^- > \mathcal{D}^+\). Consider a solution \((\Gamma', X')\) defined by \(X' = X\), \(\lambda'_j = \lambda_j\) for all \(1 \leq j \leq k\), and \(\mu'_j = \mu_i - \varepsilon\) if \(j = i\), or \(\mu'_j = \mu_j\) otherwise.
    We have
     \begin{align*}
        C_{\Gamma}(X) - C_{\Gamma'}(X) &= \sum_{s \in S\setminus S_i} |\tau(x(s)) - t(s)|  +\sum_{s \in S_i} |\tau(x(s)) - t(s)| \\& \quad - \sum_{s \in S\setminus S_i} |\tau(x(s)) - t(s)| - \sum_{s \in S_i} |\tau(x(s)) - \varepsilon - t(s)|\\
        &= \sum_{s \in S_i} |\tau(x(s)) - t(s)| - \sum_{s \in S_i} |\tau(x(s)) - \varepsilon - t(s)|\\
        &= D^+\varepsilon - D^-\varepsilon \\
        &< 0,
    \end{align*}
    using that \(\mathcal{D}^- > \mathcal{D}^+\), contradicting the optimality of \((\Gamma,X)\). Assume \(\mathcal{D}^- < \mathcal{D}^+\). This leads to a contradiction using symmetric arguments, thereby proving the claim. 
    
    Let \(s' = \argmin\{ |\Delta(s)| \mid s \in S\}\). Consider a solution \((\Gamma', X')\) defined by \(X' = X\), \(\lambda'_j = \lambda_j\) for all \(1 \leq j \leq k\), and \(\mu'_j = \mu_i + \Delta(s)\) if \(j = i\), or \(\mu'_j = \mu_j\) otherwise. As before, we can argue that 
    \(C_{\Gamma}(X) - C_{\Gamma'}(X) = D^+|\Delta(s')| - D^-|\Delta(s')| = 0,\)
    using the fact that \(D^+ = D^-.\) 
    Thus, \((\Gamma', X')\) is again an optimal solution.
    Further, we have \[\tau'(x(s')) = \tau(x(s')) + \Delta(s') = \tau(x(s')) + t(s') - \tau(x(s')) = t(s').\] Thus, for \(s' \in S_i\), we have that \(\tau'(x(s')) = t(s')\). Repeating this procedure for all relevant streams results in an optimal solution that satisfies the condition stated in the lemma. 
\end{proof}

\subsection{Algorithm }\label{sec:algorithm1}
In this section, we propose an algorithm to solve the periodic matching problem. The algorithm exploits the restricted solution space, that is implied by Lemma~\ref{lemma:interchangeOpt} and~\ref{lem:anchoring}. 

For each stream \(1 \leq i \leq k\), we let \(n_i\) denote the number of vessels that are assigned to \(i\). 
Thus, we require that \(n_1 + n_2 + \ldots + n_k = n\) and that \(1 \leq n_i \leq n\) for all \(i\). 
Further, we let \(s_i \in S\) denote the vessel that stream \(i\) is anchored to. 
For each \(1 \leq i \leq k\), we compute the periodicity with \(\lambda_i = T/n\).
The offset is computed to ensure that stream \(i\) is anchored to vessel \(s_i\), thus we first determine the nearest matching point that occurs before the arrival time of \( s_i \), that is \(
t_i = \max\{j\lambda_i \mid j\lambda_i \leq t(s_i),\, 0 \leq j \leq n_i\}.\)
The offset is then defined with \(\mu_i = t(s_i) - t_i\).
Then, compute matching points for all streams and match them to the arrivals in ascending order, resulting in matching \(X\). 
We can then compute the cost of the solution \((\Gamma, X)\).
We then determine the minimum cost by evaluating all possible combinations of \(n_1, n_2, \ldots n_k\) and \(s_1, s_2, \ldots, s_k\).

\begin{theorem}
    Algorithm 1 computes an optimal solution to the periodic matching problem in \(O(n^k)\) time.
\end{theorem}
\begin{proof}
        We discuss the time complexity and the correctness of the algorithm separately.
        \medskip 
        
        \noindent \emph{Time complexity.} 
        For \( n_1, n_2, \ldots, n_k \) satisfying \( n_1 + n_2 + \ldots + n_k = n \) and \( 1 \leq n_i \leq n \) for all \( i \), the stars and bars theorem implies there are \( O(n^k) \) possible combinations.
        For \( s_i \in S \) for each \( i \), there are \( O(n^k) \) possible combinations.
        Finally, finding the matching and computing the cost can be done in \(O(n)\) time. 
        Putting everything together, the enumeration algorithm takes \(O(n^k)\) time.

        \medskip
        \noindent \emph{Correctness.} 
        By Lemmas~\ref{lemma:interchangeOpt} and~\ref{lem:anchoring}, we know that there must exist an optimal solution that is order-preserving and anchored.
        In the algorithm, we compute the minimum over all such solutions. 
        Thus, it follows that an optimal solution is found by the algorithm.
\end{proof}

    \section{Computing optimal periodic schedules}\label{sec:scheduling}
    In this section, we develop algorithms to compute optimal periodic operating schedules for the lock. 
    We define the problem in Section~\ref{sec:problemDef2} and the existence of optimal solutions with a specific structure in Section~\ref{subsec:optSols}.
    We then proceed with the algorithms. 
    In Section~\ref{sec:opt2Streams}, we design an algorithm for the case of two streams, and in Section~\ref{sec:algorithm2} we design a dynamic programming algorithm for the general case.
    Finally, we show the existence of an approximation scheme in Section~\ref{sec:approxAlgorithm}.

    \subsection{Problem Definition}\label{sec:problemDef2}
    Again, we consider a waterway containing a single lock. 
    This lock operates in alternating cycles -- it is either actively processing or in a waiting state, during which it allows the entry and exit of vessels traveling in a specific direction. 
    We define the \emph{processing time} as the duration required for vessels to enter the lock's chamber, transition to the opposite side, and subsequently exit the chamber on the alternate side. We make the simplifying assumption that the chamber can accommodate an unlimited number of vessels. 
    Furthermore, we consider an infinite time horizon, discretized according to the processing time intervals.

    We assume that vessels arrive in streams, where each stream consists of vessels arriving from a specific direction with a constant inter-arrival rate. 
    We make the simplifying assumption that the vessels only arrive at times that align with the locket movement.
    Thus, we assume that each stream \(1 \leq i \leq k\) by its travel direction, denoted as \(\delta_i \in \{D, U\}\) (where \(D\) and \(U\) stand for downstream and upstream, respectively), and its periodicity \(\lambda_i\) and its offset \(\mu_i\). 
    We assume that \(\lambda_i, \mu_i \in \mathbb{Z}_+\) and \(\lambda \geq 1\) and \(1\leq  \mu_i \leq \lambda_i\).
    The arrival times of the vessels of stream \(i\) are then given by the sequence  \((\mu_i + j\lambda_i)_{j\in \mathbb{Z}_+}\).
    Finally, let \(I_D\) denote all streams with a downstream travel direction, while \(I_U\) denotes those streams traveling upstream, and let \(I = I_D \cup I_U\).
    
    At times when the lock is not processing, it is aligned either to downstream or upstream. If the lock is aligned downstream, the lock operator has two distinct actions available. The first action, denoted by \(D_a\) is changing the lock's alignment from downstream to upstream while processing a batch of vessels waiting on the downstream side. The second action is to remain at its current alignment for a time period. Similarly, when the lock is aligned upstream, the lock operator can choose from two distinct actions. The first action, denoted by \(U_a\) is changing the lock's alignment from upstream to downstream while processing a batch of vessels waiting on the upstream side. The second action is to remain at its current alignment for a time period. 
    We define \( f(D) = U \) and \( f(U) = D \).

    A \emph{schedule} \(\sigma\) for the lock operator is defined as a sequence \(\{\sigma(t)\}_{t=1}^{\infty}\), where \(\sigma(t)\) denotes the action of the lock operator during period \(t\), that is \(\sigma(t) \in \{D, U, W\}\). 
    We assume that the orientation of the lock at the beginning of the time horizon is arbitrary.
    A \emph{feasible} schedule requires that the lock is alternating between directions \(D\) and \(U\). 
    This means if \(U\) is chosen at time period \(t\), then the most recent action prior to \(t\) that is not \(W\) must be \(D\) and vice versa. 
    A sequence \(s\) is \emph{periodic} with period \(T\) if it consists of repetitions of a finite sequence of length \(T\). 
    That is, there exists a \(T \in \mathbb{Z}_+\) such that \(s(t) = s(t+T)\) for all \(t\). 

    For each stream \(1 \leq i \leq k\), we define its arrival pattern with 
     \[
        a_i(t) = \begin{cases}
            1 &\text{ if } t \equiv \mu_i \mod \lambda_i,\\
            0 &\text{ otherwise.}
        \end{cases}
    \]
    for each time period \(t\). 
    The arrival patterns for all downstream and upstream streams, denoted by \(a_D\) and \(a_U\) respectively, are given by 
    \[
        a_\delta(t) = \sum_{i \in I_\delta} a_i(t), \text{ for } \delta\in \{D,U\}
    \]
    Finally, we let \(a(t) = (a_D(t), a_U(t))\) for all \(t\), representing the combined arrival patterns for both directions at time \(t\). 
    Furthermore, we define the \emph{waiting time} of a vessel as the time elapsed between its arrival and when it is processed. 
    
    Given a schedule \(\sigma\) and a time \(t\), we define \(n_D(t)\) as the number of vessels present at the lock in the downstream direction and \(n_U(t)\) as the number of vessels in the upstream direction. 
    For notational convenience, we define \(n_D(0) = n_U(0) = 0\), and for \(t > 1\), we define
    \[n_\delta(t) =
    \begin{cases}
        0, & \text{if } \sigma(t) = \delta, \\
        n_\delta(t-1) + a_\delta(t), & \text{otherwise.}
    \end{cases}\]
    We define the waiting time per time period \(t\) with \(C_\sigma(t) = n_D(t) + n_U(t)\) and we define the \emph{long-run average waiting time} by 
    \[C_{\text{avg}, \sigma} = \lim_{T \to \infty} \frac{1}{T} \sum_{t=1}^{T} C_\sigma(t).\]
    Our aim is to minimize long-run average waiting time over all feasible schedules.
    We term the problem of finding a periodic and optimal schedule as the \emph{periodic lock scheduling problem}. 

    \subsection{Characterizing optimal solutions} \label{subsec:optSols}
    Before presenting the algorithms, we first establish a series of structural results that show the existence of optimal solutions with a specific structure and that are periodic.
    
    \begin{lemma}\label{lemma:cyclicArrival}
        It holds that \(a(t) = a(t + \Lambda)\) with \(\Lambda = \lcm(\lambda_1, \ldots, \lambda_k)\). In other words, the arrival pattern is periodic with period \(\Lambda\).
    \end{lemma}
    \begin{proof}
    Consider stream \(i\) with its arrival pattern \(a_i\), and let \(\Lambda = \lcm(\lambda_1, \ldots, \lambda_k)\). 
    By definition of \(\Lambda\), we know that \(t \equiv t + \Lambda \bmod \lambda_i\) and thus \(a_i(t) = a_i(t + \Lambda)\). By applying the periodicity of each \(a_i(t)\), we know that 
    \[
        a_\delta(t) = \sum_{i \in I_\delta} a_i(t) = \sum_{i \in I_\delta} a_i(t + \Lambda) = a_\delta(t + \Lambda), \text{ for } \delta\in \{D,U\}.
    \]
    Thus, we have that \(a(t) = (a_D(t),a_U(t)) = (a_D(t+\Lambda),a_U(t+\Lambda)) = a(t+\Lambda)\), thereby showing the lemma.
    \end{proof}

    \begin{lemma}\label{lemma:existence}
        There exists a schedule \(\sigma\), such that \(C_{\text{avg}, \sigma}\) is well-defined and finite.
    \end{lemma}
    \begin{proof}
        Consider a schedule \(\sigma\) that is periodic with period \(k \in \mathbb{Z}_+\) with at least one \(D\) operation and at least one \(U\) operation. 
        By Lemma~\ref{lemma:cyclicArrival}, we know that \(a(t) = a(t+\Lambda)\) and let \(M = \lcm(\Lambda, k)\). 
        By definition of the waiting time, we know that \(C_\sigma(t) = C_\sigma(t + M)\) for all \(t\) and we define \(B \coloneqq \sum_{t =1}^M C_\sigma(t) \), where \(B \in \mathbb{Z}_+\). Let \(T \in \mathbb{Z}_+\) and \(q \coloneqq \floor{T/M}\) and consider
        \[\frac{1}{T} \sum_{t=1}^{T} C_\sigma(t) = \frac{q}{T} \sum_{t=1}^{M} C_\sigma(t) + \frac{1}{T} \sum_{t=qM+1}^{T} C_\sigma(t).\]
        Observe that \((q/T)\cdot \sum_{t=1}^{M} C_\sigma(t) \to B/M\) as \(T \to \infty\). Furthermore, we know that \(0 \leq \sum_{t=qM+1}^{T} C_\sigma(t) \leq B\), implying that \((1/T)\cdot \sum_{t=qM+1}^{T} C_\sigma(t) \to 0\) as \(T\to \infty.\)
        Thus, we know that 
      \[C_{\text{avg}, \sigma} = \lim_{T \to \infty} \frac{1}{T} \sum_{t=1}^{T} C_\sigma(t) = \frac{B}{M},\]
      thereby showing the claim.

    \end{proof}

     \begin{lemma}\label{lemma:2wOpt}
        There exists an optimal schedule \(\sigma\) such that no two consecutive periods \(t\) and \(t+1\) both have the action \(\sigma(t) = \sigma(t+1) = W\). In other words, the schedule is single-wait.
    \end{lemma}
    \begin{proof}
    Consider an optimal solution \(\sigma\). 
    By Lemma~\ref{lemma:existence}, we know that such a solution exists and \(C_{\text{avg}, \sigma}\) is finite.
    Consider a time period \(t'\), such that \(\sigma(t' + i) = W\) for some \(t'\) and \(i = 1, 2, \ldots, T\) for some \(T \geq 2.\)
    If such a time period does not exist, we are done.
    We may assume that the lock is oriented downstream at the beginning of time \(t' + 1\). 
    Consider schedule \(\sigma'\), defined for all \(t\) with
    \[\sigma'(t)=
    \begin{cases}
    D & \text{ if } t=t'+1, \\
    U & \text{ if } t=t'+2, \\
    \sigma(t) & \text{otherwise.}
    \end{cases}
    \]
    Observe that \(\sigma'\) is alternating and is oriented downstream at the beginning of time \(t'+3\). 
    Therefore, the schedule is feasible.
    We repeat this procedure for all relevant time periods and denote the resulting schedule with \(\sigma'\).
    Moreover, since \(C_\sigma(t) \geq C_{\sigma'}(t)\geq 0\) for all \(t\), it follows that \(C_{\text{avg}, \sigma'}\) is also well-defined and finite. Further, it implies that 
    \[C_{\text{avg}, \sigma} = \lim_{T \to \infty} \frac{1}{T} \sum_{t=1}^{T} C_\sigma(t) \geq  \lim_{T \to \infty} \frac{1}{T} \sum_{t=1}^{T} C_{\sigma'}(t) = C_{\text{avg}, \sigma'}.\]
    Thus, \(\sigma'\) must also be an optimal solution, thereby showing the claim.
    \end{proof}

    Given a schedule \(\sigma\), we define a state space that represents the possible states of the arrival pattern and the lock. 
    Given the beginning of time period \(t\), the lock's alignment is denoted by \(\delta(t)\). 
    Define \(w_\delta(t)\) as the count of waiting operations \(W\) on the \( \delta \) alignment since the most recent switch from the opposite alignment.
    Symmetrically, define \(w_{f(\delta)}(t)\) as the count of waiting operations \(W\) on the \( f(\delta) \) alignment since the most recent switch from the opposite alignment.
    We define the states at time \(t\) with \(s(t) = (a(t), w_\delta(t), w_{f(\delta)}(t), \delta(t))\).
    
    \begin{lemma}\label{lemma:cyclicOpt}
        There exists an optimal schedule \(\sigma\) that is periodic with period of length at most \(8\Lambda\), where \(\Lambda = \lcm(\lambda_1, \ldots, \lambda_k)\).
    \end{lemma}

    \begin{proof}
        Consider an optimal schedule \(\sigma\).
        By Lemma~\ref{lemma:2wOpt}, we may assume that \(\sigma\) is single-wait. Given time \( t \), observe that by Lemma~\ref{lemma:cyclicArrival}, there are \(\Lambda\) possible values for \( a(t) \), while \( w_\delta(t) \), \( w_{f(\delta)}(t) \), and \( \delta(t) \) each have two possible values.
        Thus, there are \(8\Lambda\) possible states for \(s(t)\). 

        We partition \(\sigma\) into windows of \(8\Lambda + 1\) time steps. Due to the pigeonhole principle, each window must contain a state that occurs (at least) twice. For each window, identify two occurrences of such a repeated state and the sequence of actions between the two occurrences of that state. Note that this sequence of actions is of length at most \(8\Lambda\).
        
        Because there are only finitely many such sequences, there must exist a state \(s\) and sequence of actions between two occurrences of state \(s\) that occurs in a non-vanishing (as \(T\to\infty \)) proportion of all windows. We claim that the periodic schedule obtained by repeating this sequence of actions is optimal.
        
        Otherwise, if the long-term average cost of this periodic schedule is higher, this contradicts the optimality of \(\sigma\): every time in \(\sigma\) that we encounter this sequence of actions we could instead not perform those actions, and jump ahead to after the next occurrence of state \(s\) and perform those actions instead. Since the average cost of the skipped actions is higher than the long-term average cost of \(\sigma\), and we skip a non-vanishing proportion of actions, the long-term average cost of the not-skipped actions must be lower than the long-term average cost of \(\sigma\).
    \end{proof}
    
    \subsection{Closed-form expression for two streams}\label{sec:opt2Streams}
    In this section, we present a solution to the problem while restricting the instances to two streams.
    We assume one stream is in the downstream direction, while the other is in the upstream direction. 
    Thus, the input instance reduces to four parameters \(\mu_D\), \(\mu_U\), \(\lambda_D\) and \(\lambda_U\). 
    The case with both streams in the same direction is solved with symmetric arguments. 

    First, we consider the case where one of the periodicities is equal to one. We may assume that \(\lambda_D = 1\). 
    \begin{lemma}
        There exists an optimal schedule \(\sigma\) that is periodic with period 2. Specifically, we have that
       \[
        (\sigma(1), \sigma(2)) = 
        \begin{cases} 
        (D, U), & \text{if } \mu_U \text{ is odd,} \\
        (U, D), & \text{otherwise.}
        \end{cases}
        \]
    \end{lemma}
    \begin{proof}
        Observe that \(\lambda_D = 1\) implies that in any schedule, we have a waiting cost from vessels arriving downstream at least every two time units.
        If \(\lambda_U\) is even and \(\mu_D\) is odd, consider the periodic schedule repeating \((D, U)\). 
        It is easy to verify, that the occurs one waiting cost cost every two time units of vessels arriving downstream and no waiting cost of vessels arriving upstream. 
        Thus, the schedule must be optimal. 
        If \(\mu_D\) is even, symmetric reasoning show that the periodic schedule repeating \((U, D)\) is optimal.

        Assume \(\lambda_U\) is odd and consider an arbitrary schedule \(\sigma\). 
        Observe that inserting a waiting between any operation, may decrease unit of waiting cost of vessels arriving in upstream direction and always increase a unit of waiting cost of vessels arriving in downstream direction.
        Therefore, there must exist an optimal schedule, that is wait-free. 
        Observe that in both the periodic schedule repeating \((U, D)\) and \((D, U)\), there occurs one waiting cost every two time units of vessels arriving downstream and one waiting cost every \(2\lambda_U\) time units of vessels arriving upstream.
        Thus, both schedules are optimal, thereby showing the claim.
    \end{proof}
    
    For the remainder of this section, we assume that we are in the case where \(\lambda_D \geq 2\) and \(\lambda_U \geq 2\).
    \begin{lemma}\label{lemma:lowerBound}
        For any optimal schedule \(\sigma\) it holds that
        \[
        C_{avg, \sigma} \geq \begin{cases} \frac{1}{\Lambda} & \text{if } \mu_U-\mu_D \equiv 0 \quad \mod \gcd\left(\lambda_D, \lambda_U\right)\\
        0 & \text{otherwise},
        \end{cases}
        \]
        where \(\Lambda = \lcm(\lambda_D, \lambda_U)\).
    \end{lemma}
    \begin{proof}
        Consider an instance of the periodic lock scheduling problem with two streams. 
        Since \(\lambda_D \geq 2\) and \(\lambda_U \geq 2\), in an optimal solution a waiting cost occurs only when arrivals from both streams coincide at a time period.
        Further, we know that two arrivals coincide at each time period \(t\) if and only if \(t\) satisfies the following linear Diophantine equation
         \begin{align}
        t = \mu_D + k_D\lambda_D = \mu_U + k_U\lambda_U. \label{eq:linDioEq}
        \end{align}
        for any \(k_D,k_U\in \mathbb{Z}_+\). A linear Diophantine equation is of the form \(ax + by = c\), where \(a\), \(b\), and \(c\) are given integers. According to basic number theory, this Diophantine equation has an integer solution for \(x\) and \(y\) if and only if \(c\) is divisible by \(\gcd(a, b)\). Moreover, if \((x^0, y^0)\) represents a solution, then any other solution can be expressed as \((x^0 + mv, y^0 - mu)\), where \(m\) is an arbitrary integer, and \(u\) and \(v\) are the respective quotients of \(a\) and \(b\) divided by \(\gcd(a, b)\). For more details, we refer to Hardy and Wright's book \cite{hardy1979introduction}.

        Hence, for Equation~\eqref{eq:linDioEq}, simple algebra shows that an integer solution exists if and only if 
        \[\mu_U - \mu_D \equiv 0 \mod \gcd(\lambda_D, \lambda_U)\]
        For a particular solution \((k_D^0, k_U^0)\) and some \(m \in \mathbb{Z}_+\), all the solutions of \(k_D\) can be expressed as
        \begin{align*}
        k_D &= k_D^0 + m\frac{\lambda_U}{\gcd(\lambda_D, \lambda_U)}.
        \end{align*}
        Therefore, it follows that for any \(m \in \mathbb{Z}_+\), the time steps with occurrences of two arrivals are characterized by
        \begin{align*}
        \mu_D + k_D\lambda_D &= \mu_D + k_D^0\lambda_D + m\frac{\lambda_D\lambda_U}{\gcd(\lambda_D, \lambda_U)} 
        \\ &= \mu_D + k_D^0\lambda_D + m\lcm(\lambda_D, \lambda_U)
        \\ &= \mu_D + k_D^0\lambda_D + m\Lambda,
        \end{align*}
        thereby showing the lemma.
    \end{proof}
    As a representation of the solution, we present a closed-form expression that determines the single action for a given time period. We let \(\delta_1\) and \(\delta_2\) denote directions, such that \(\lambda_{\delta_1} \leq \lambda_{\delta_2}\), and let \(A_1\) denote the action of processing vessels from stream direction \(\delta_1\) and let \(A_2\) denote the action of processing vessels from stream direction \(\delta_2\). Then 
    \begin{equation}\label{eq:closedExpr}
        \sigma(t) = 
        \begin{cases} 
            A_1, & \text{if } C_1(t), \\
            A_2, & \text{otherwise, if } C_2(t) \text{ or } C_3(t) \text{ or } C_4(t), \\
            W, & \text{otherwise},
        \end{cases}
    \end{equation}
    where \(C_1(t), C_2(t), C_3(t)\), and \(C_4(t)\) are boolean functions that return {\sc True} if their respective expressions are satisfied, and {\sc False} otherwise. These functions are defined as follows:
    \begin{align*}
     C_1(t) &\coloneqq \left(t \equiv \mu_{\delta_1} \mod \lambda_{\delta_1}\right) \text{ and  not } \left( (t-1) \equiv \mu_{\delta_1} \mod \lambda_{\delta_1}\right)\\
     C_2(t) &\coloneqq \left((t-1) \equiv \mu_{\delta_1} \mod \lambda_{\delta_1}\right) \text{ and } \left((t-1) \equiv \mu_{\delta_2} \mod \lambda_{\delta_2}\right) \\
    C_3(t) &\coloneqq \left(t \not\equiv \mu_{\delta_1} \mod \lambda_{\delta_1}\right) \text{ and } \left(t \equiv \mu_{\delta_2} \mod \lambda_{\delta_2}\right) \\
    C_4(t) &\coloneqq \left((t + 1) \equiv \mu_{\delta_1} \mod \lambda_{\delta_1}\right) \text{ and }\left(\mu_{\delta_1} + \left\lfloor\frac{t-\mu_{\delta_1}}{\lambda_{\delta_1}}\right\rfloor\lambda_{\delta_1} > \mu_{\delta_2} + \left\lfloor\frac{t-\mu_{\delta_2}}{\lambda_{\delta_2}}\right\rfloor\lambda_{\delta_2} \right).
   \end{align*}
   Condition \(C_1(t)\) ensures that each arrival from stream \(\delta_1\) is processed immediately, without any waiting time. 
   In situations where two arrivals occur at time \(t-1\), condition \(C_2(t)\) guarantees that the vessel from stream \(\delta_2\) is processed in the next time period, \(t\). 
   For scenarios with an arrival from stream \(\delta_2\) but none from stream \(\delta_1\), condition \(C_3(t)\) ensures that the stream \(\delta_2\) arrival is processed.
   Lastly, in condition \(C_4(t)\), the lock is reoriented to direction \(\delta_1\) if it is currently oriented to \(\delta_2\) and an arrival from \(\delta_1\) occurs in the next time period. 

  Given an input instance restricted to two streams \(\mathcal{I} = (\mu_D, \mu_U, \lambda_D, \lambda_U)\) and a given time \(t\), a lock operator can evaluate Equation~\eqref{eq:closedExpr} in polynomial time to determine the optimal operation for time \(t\). Note that it is not possible to compute the entire schedule in polynomial time since the length of the schedule and thus the size of the output could be exponential in the size of the input.
  Denote by \(\sigma^*\) the schedule obtained by repeatedly applying Equation~\eqref{eq:closedExpr}.
   By Lemma~\ref{lemma:cyclicArrival}, we know that \(\sigma^*\) is periodic with period \(\Lambda = \text{lcm}(\lambda_U, \lambda_D)\).

   \begin{theorem}
       \(\sigma^*\) is an optimal schedule. Furthermore, \(\sigma^*(t)\) can be computed in polynomial time.
   \end{theorem}
   \begin{proof}
            We discuss the time complexity and the correctness of the algorithm separately.
        \medskip 
        
        \noindent \emph{Time complexity.} The computation of \(\sigma^*(t)\) involves only a constant number of arithmetic operations. These can clearly be performed in time polynomial in the bit lengths of $t$ and $\mathcal{I}$.
        \medskip 
        
        \noindent \emph{Correctness.} Based on the previous discussion, it is straightforward to verify that the cases are not coinciding and that the lock is alternating.
        Given time \(t\), if there is an arrival only from one stream, then \(C_1(t)\) and \(C_3(t)\) ensure that no waiting time is incurred at any arrival.
        If there are two arrivals, then \(C_1(t)\) and \(C_2(t)\) imply that one unit of waiting time is incurred.
        Thus, if there are time periods with coinciding arrivals, there is at most one waiting time every \(\Lambda = \lcm(\lambda_D, \lambda_U)\) time units.
        Otherwise, there is no waiting time.
        Thus, by Lemma~\ref{lemma:lowerBound}, the schedule is optimal.
   \end{proof}

    \subsection{Algorithm}\label{sec:algorithm2}

    We design a dynamic programming algorithm that computes an optimal periodic schedule in pseudo-polynomial time.

    We now construct the dynamic programming table. 
    By Lemma~\ref{lemma:2wOpt}, we know that the lock can wait for at most one time period, before alternating to the other side.
    Thus, we know that the possible states of the lock at time \(t\) are defined by the set 
    \[\mathcal{S} = \{(\delta, w_\delta, w_{f(\delta)}) \mid  \delta \in \{D, U\}, \, w_\delta \in \{0, 1\}, \, w_{f(\delta)} \in \{0, 1\}\}, \]
    where in a state \(S\in \mathcal{S}\) the lock is positioned on side \( \delta \), there are \( w_\delta \) waiting operations on side \( \delta \), and \( w_{f(\delta)} \) waiting operations on side \( f(\delta) \).
    Let \(D(t, S, S_0)\) represent the minimum waiting time of lock operations for time periods \(1,2, \ldots, t\), while the lock is in state \(S\) at time \(t\) and was at state \(S_0\) at time \(1\). By Lemma~\ref{lemma:cyclicOpt}, we restrict our search for schedules that are of length \(T \coloneqq 8\Lambda\), where \(\Lambda = lcm(\lambda_1, \ldots, \lambda_k)\). Thus, we restrict the time periods to the interval \(1 \leq t \leq T\).
    
    The state space is initiated for all \(S, S_0 \in \mathcal{S}\) with 
    \begin{equation}\label{eq:init}
    D(1, S, S_0) =
    \begin{cases}
    0, & \text{if } S = S_0, \\
    \infty, & \text{otherwise}.
    \end{cases}
    \end{equation}
    The dynamic programming table is updated by the following recursive relation. 
	We proceed in lexicographic order. 
    Assume all states of the dynamic programming table up to but not including \((t, S, S_0)\) are filled in. 
    Given a time t, the lock can either wait for one time period, provided it has not already waited on the current side since the last switch, or it can proceed to process a lockage. Thus, given a state \(S\), all potential preceding states form a set
    \[
    pred(S) =
    \begin{cases}
    \{(f(\delta), w_{f(\delta)}, 0), (f(\delta), w_{f(\delta)}, 1)\}, & \text{if } w_\delta = 0, \\
   \{(\delta, w_\delta - 1, w_{f(\delta)})\}, & \text{otherwise}.
    \end{cases}
    \]

    Consider time \(t\), states \(S\) and \(S'\in pred(S)\). If \(w_{\delta} = 1\), then a waiting operation was scheduled at time \(t-1\) and no waiting cost occurs.
    Otherwise, the lock switched from direction \(\delta'\) to \(\delta\) at time \(t-1\). Thus, all vessels that arrived from direction \(\delta'\) during the time periods \(t-1, \ldots, t-w(S')\), where \(w(S') = 2 + w'_{\delta} + w'_{f(\delta)}. \) are processed.
    We define the arrival pattern cyclically, such that \( a_\delta(-i) = a_\delta(T - i) \) for any \( i \geq 0 \) and \(\delta \in \{D,U\}\).
    It follows that the transition cost between the two states is defined with
    \[
    C(t, S', S) =
    \begin{cases}
    \sum_{i=1}^{w(S')} (i-1)a_{\delta'}(t - i), & \text{if } w_{\delta} = 0, \\
    0, & \text{otherwise}.
    \end{cases}
    \]

    Then, the minimum waiting scheduling all operations for the time periods \(1, 2, \ldots, t\), while the lock is in state \(S\) at time \(t\) and was at state \(S_0\) at time \(0\) is determined by 
	\begin{equation}\label{eq:Recursion}
		D(t, S, S_0) = \min_{S' \in prec(S)} D(t-1, S', S_0) + C(t, S',S)
	\end{equation}
    Finally, the minimum long-run average waiting time is computed by executing the value at the maximum length of the schedule for each initial state of the lock, that is
    \begin{equation}\label{eq:Termination} \frac{1}{T}\min\{D(T, S, S_0) + C(1, S,S_0) \mid S\in pred(S_0), S_0 \in \mathcal{S} \} \end{equation}

    The dynamic programming algorithm is defined by its initialization in Equation~\eqref{eq:init}, its recursion in Equation~\eqref{eq:Recursion} and by returning the minimum long-run average waiting time as specified in Equation~\eqref{eq:Termination}. 
    We refer to Equations~\eqref{eq:init},~\eqref{eq:Recursion} and~\eqref{eq:Termination} as Algorithm~2. 
	\begin{theorem}\label{thm:TSPNoProcDeadlines}
		Algorithm~2 solves the periodic scheduling problem in \(O(kT)\) time.
	\end{theorem}
    \begin{proof}
        We discuss the time complexity and the correctness of the algorithm separately.
		\medskip 
		
		\noindent \emph{Time complexity.} 
        Observe that set \(\mathcal{S}\) is of \(O(1)\) size and thus the size of the dynamic programming table is bounded by \(O(T)\). 
        In each state, the computation of the cost can be done in \(O(k)\) time, while there are \(O(1)\) transitions. 
        Therefore, the algorithm runs in \(O(kT)\) time.
		
		\medskip
		\noindent \emph{Correctness.} 
        Due to the construction of the dynamic programming table, the algorithm enumerates all periodic schedules with a period \(T\) that are single-wait, while pruning schedules that are dominated.
        By Lemma~\ref{lemma:2wOpt} and~\ref{lemma:cyclicOpt}, there must exist an optimal schedule that is periodic with period \(T\) and single-wait.
        Thus, the algorithm must return an optimal solution.
    \end{proof}

    \subsection{Approximation scheme}\label{sec:approxAlgorithm}
    The dynamic programming approach computes an optimal periodic schedule of length that is bounded by a polynomial of the input parameters in unary encoding. 
    If these parameters are very large, it is not only intractable to compute such a schedule, but even storing such a large schedule would be impractical. 
    For practical purposes, it would be useful to have an algorithm that, in reasonable time, computes a good sequence of actions for the immediate future. 
    In this section, we show how to use Algorithm~2 to only look a "short" amount of time into the future to obtain schedules that are close to optimal in a reasonable amount of time.

    We design a so-called incremental polynomial time algorithm, which means that the next set in the list of output sets is generated in time that is polynomial in the size of the input plus the size of the generated part of the output. Such algorithms have been applied in enumeration and high multiplicity scheduling, see e.g. Brauner et al. \cite{brauner2007multiplicity} and Golovach et al. \cite{golovach2015incremental}. 
    Our algorithm yields schedules with an approximation guarantee. 
    For any \(\varepsilon > 0\) and a given time \(t\), we compute a schedule for the next \(O(k^2 / \varepsilon)\) time periods, which is an \((1 + \varepsilon)\)-approximation of the optimal periodic schedule over this interval.
    Our algorithm is thus an approximation scheme that runs in incremental polynomial time.
    Besides the theoretical interest, the incremental construction of the schedule may be used by practitioners to adjust their schedule to unforeseen delays of some of the vessel streams.

    Before discussing the algorithm, we require two subroutines.  
    Consider two time periods \(t\) and \(t'\), where \(t \leq t'\).
    Define the algorithm \({ALG}(t, t')\), which computes the optimal schedule for the time interval from \(t\) to \(t'\), assuming the lock's position at time \(t\) is \emph{free}.
    Here, \emph{free} means that the lock's position can be arbitrarily oriented using lock operations during time periods \(t-2\) and \(t-1\).
    Additionally, define the algorithm \({ALG}(t, t', \delta)\) that computes the optimal cost and schedule from time \(t\) to \(t'\), assuming the lock's position at time \(t\) is \(\delta\). 
    These algorithms are straightforward extensions of Algorithm~2 and are omitted from the writeup. 
    For a given schedule \(\sigma\), we define \(\sigma(t,t')\) as the schedule from \(t\) to \(t'\) and \(c(\sigma(t,t'))\) as the waiting time that occurs in that period.

    The algorithm takes as input a time period \(t\), a constant \(T\), and optionally a lock position \(\delta\) at time \(t\); otherwise, the lock position at time \(t\) is assumed to be free. 
    The algorithm outputs a schedule \(\sigma(t,t'))\), and optionally a lock position \(\delta\) at time \(t'+1\). If the position \(\delta\) is given, the algorithm computes \(ALG(t, t + T, \delta)\); otherwise, it computes \(ALG(t, t + T)\). 
    In either case, let \(\sigma(t,t+T)\) denote the resulting schedule.
    The algorithm evaluates the following cases:
    
\begin{itemize}
    \item[(i)] If, during the interval from \(t\) to \(t + T\), there are consecutive periods without arrivals, let \(t_1\) denote the last time period before these intervals and \(t_2\) denote the first time period after them. We define a schedule \(\sigma'\) such that \(\sigma'(t') = \sigma(t')\) for \(t \leq t' \leq t_2 - 2\), while \(\sigma'(t_2 - 1)\) and \(\sigma'(t_2)\) are left undefined to ensure a free lock position at the subsequent time period.
    Return \(\sigma'(t, t_2)\)
    
    \item[(ii)] Otherwise, and if \(c(\sigma(t, t + T)\leq 2k/\varepsilon\), let \(t' = t + T/2 + 1\). 
    Let \(\delta'\) denote the position of the lock at time \(t'+1\). 
    Return \(\sigma(t, t')\) and \(\delta'\). 
    
    \item[(iii)] Otherwise, we define a schedule \(\sigma'\) such that \(\sigma'(t') = \sigma(t')\) for \(t \leq t' \leq t + T\), while \(\sigma'(t + T + 1)\) and \(\sigma'(t + T + 2)\) are left undefined to ensure a free lock position at the subsequent time period.
    Return \(\sigma'(t, t + T + 2)\).
\end{itemize}
    We refer to this as Algorithm~3 and now state the approximation guarantee. 
    \begin{theorem}
        Given a time period \(t\), assume that the lock is in an arbitrary position or in the position \(\delta\), such that there exists an optimal periodic solution with the lock positioned at \(\delta\) at time~\(t\). 
        Consider the output of Algorithm~3 when called with \(t\), \(T = 40k^2/\varepsilon\), and optionally \(\delta\).
        We have that \(\sigma(t,t')\) such that \(c(\sigma(t,t')) \leq (1+\varepsilon)c(\sigma^*(t,t'))\), where \(\sigma^*\) is an optimal periodic solution.
        If a position \(\delta\) is part of the output, then an optimal periodic solution exists with the lock positioned at \(\delta\) at time \(t' + 1\).
    \end{theorem}
    \begin{proof}
        We show that the statement holds true for each case that is defined in Algorithm~3. 

        Consider case (i). Since there there are at least two consecutive periods of no arrival. We know that the lock can changed in any position during time \(t_1\) and \(t_2\) with no additional waiting time. Since, the position of the lock at \(t\) is either arbitrary or the same position of the optimal periodic schedule, we know that \(c(\sigma(t,t_2-2)) \leq c(\sigma^*(t,t_2-2))\) for an optimal periodic  solution \(\sigma^*\). 
        Furthermore, we can use time period \(t_2 - 2\) and \(t_2 - 1\) to position the lock arbitrarily at no cost.

        Consider case (ii). We know that during the interval from \(t\) to \(t + T\), there are no consecutive periods without arrivals, and \(c(\sigma(t, t + T)) \leq 2k/{\varepsilon}\). 
        Given that the length of the time window is \(40k^2/{\varepsilon}\), we know that in the second half of the solution, there is a consecutive period of at least \(10k\) with zero cost.
        We claim that the optimal periodic  solution aligns with the partial exact solution at some point during this zero-cost interval. If it does not align at some point, it implies that all arrivals during this period wait at least one unit of time. 
        Since we have at least one arrival every three time units, this would incur a cost of at least \(5k\). 
        In this scenario, we could improve the optimal periodic solution by switching to the partial exact solution at the start of the zero-cost period and switching back at the end. 
        Given that the cost of a switch is upper bounded by \(2k\), the two switches would incur a total cost of \(4k\), thereby reducing the cost of the optimal solution by at least \(k\), implying a contradiction. Since the partial exact solution reaches the same state as the optimal periodic  solution at some point, we know that the cost of the partial exact solution leading up to this point is also optimal. Thus, \(c(\sigma(t,T/2)) \leq c(\sigma^*(t,T/2))\) and the position at \(T/2+1\) is the same as the one from the optimal periodic solution, thus satisfying the claimed property.

        Consider case (iii). We know that during the interval from \(t\) to \(t + T\), there is no consecutive period without arrivals, and \(c(\sigma'(t, t + T) > 2k/{\varepsilon}\). We can use time periods $T+1$ and $T+2$ to set the lock at an arbitrary position \(\delta\) at cost of at most $2k$. Since the \(\sigma'\) solution is a lower bound on the cost of the part of the optimal periodic  solution from time $t$ to time $t + T$, we know that 
        \begin{align*}
            c(\sigma'(t, t+T+2))
            &= c(\sigma'(t, t+T)) + c(\sigma'(t+T, t+T+2)) \\
            &\leq \frac{2k}{\varepsilon} + 2k \\
            &\leq (1 + \varepsilon) c(\sigma^*(t, t+T+2))\\        \end{align*}
        thus satisfying the claimed property.
    \end{proof}
    
    \section{Numerical experiments} \label{sec:numerical}
    Based on the algorithms introduced in Section~\ref{sec:matching} and~\ref{sec:scheduling}, this section aims at deriving first algorithmic and managerial insights regarding periodic lock management.
    As outlined above, further research is necessary to fully understand the system in its general form, and thus we focus here on a fairly basic research question.

    \vspace{.1cm}
    \begin{itemize}
    \item[Q1:] \hypertarget{Q3}{} How computationally expensive is it to compute the optimal periodic matching, and how well does the periodic arrival align with the actual data?
    \item[Q2:] \hypertarget{Q4}{}  How good is the optimal periodic schedule when applied to periodic arrivals, compared to its performance on the original arrival data?
    \item[Q3:] \hypertarget{Q5}{} When using the optimal periodic matching as input to create an optimal schedule, how does its performance on the original data compare to other intuitive scheduling policies?
    \end{itemize}
    \vspace{.1cm}

    The numerical experiments were conducted using C++20 on a MacBook Air (2022) equipped with an Apple M2 chip, 16 GB of RAM, and running macOS Sonoma 14.4.1. In Algorithm~1, a recursion was implemented to enumerate all possible distributions of vessel counts partitioned across all streams.
    The recursive algorithm was optimized by pruning branches that lead to symmetric distributions, resulting in a substantial speedup. Furthermore, the dynamic programming approach that constitutes Algorithm~2 was implemented using tabulation. 

    \subsection{Experimental Setup}
    We analyzed AIS data containing GPS and timestamp information to determine vessel arrival times at the Prinses Marijkesluizen, a key connection between the Rhine and the Amsterdam-Rhine Canal in the Netherlands. Timestamps are discretized into one-minute intervals. 
    The dataset captures vessel arrivals at the Prinses Marijkesluizen between January 2, 2019, and March 28, 2019 spanning 86 days.
    On average, the upstream direction experiences 53.41 vessel arrivals per day, while the downstream direction sees 55.01 arrivals per day. 
    Additionally, based on GPS coordinates, the average duration of a lock operation was calculated and rounded to the nearest integer, resulting in a duration of 21 minutes.

    For the periodic matching problem, for each day and each direction we select a number of stream \(k \in \{2,3,4\}\) and a number of vessels \(n \in \{30, 40, 50\}\) on each direction. 
    We define the input by specifying \(k\), the arrival times of the first \(n\) vessels, and defining the time interval by the start of the day and the arrival time of the \(n\)-th vessel. 
    If fewer than \(n\) vessels arrive on that day, we set \(n\) to the actual number of arriving vessels.
    Thus, Algorithm~1 is applied to 1548 input instances.
    For a given day and set of parameters, the outputs of Algorithm~1 for each direction serve as input instances for Algorithm~2. 
    This yields a total of 774 input instances for Algorithm~2.
    
    We further implemented intuitive operational strategies for managing the lock operations.
    First, we evaluate a strategy based on alternating upstream and downstream lock movements. Two schedules are generated, each starting with a different initial movement, and the better-performing schedule is selected.
    We refer to this strategy as \emph{alternating}.
    Second, we implement a first-in-first-out, also \emph{FIFO}, strategy, where the lock operator considers only the current and past time periods. 
    If a vessel arrives during the current time step, the operator executes a lock operation. Similarly, if vessels arrived in the previous time step but were not processed, the operator also executes a lock operation.
    Lastly, we implemented an advanced first-in-first-out, also \emph{advFIFO}, strategy, where the lock operator also considers the next time period. If no vessels arrive in the current time period, but a vessel arrives from a specific side in the next time period, the lock is preemptively operated to accommodate the incoming vessel.   
    Ties are broken arbitrarily.
    
    \subsection{Evaluation of research question Q1}
    In order to analyse Q1, we applied Algorithm~1 on the dataset.
    The results are summarised in Table~\ref{tab:resultsAlg1}.
    Grouped by parameters \(k\) and \(n\), the table provides information about runtime (column \emph{runtime}), measured in seconds, and the average deviation between each arrival and its corresponding matched time, measured in minutes (column \emph{Fit}).
    
    As we would expect from an \(O(n^k)\) algorithm, the running time increases quickly for larger \(n\) and \(k\) having a runtime of less than a second for \(k = 2\) and \(n = 20\) up to \(49\) minutes and \(50\) seconds for \(k = 4\) and \(n = 50\). 
    Furthermore,  the fit improves as the number of streams increases, which follows from the fact that any solution \(k-1\) streams can be replicated with \(k\) streams.
    The trend in fit with respect to the number of vessels is less clear.
    For all \(k\), we observe that the fit increases from \(n = 20\) to \(n = 30\), then it decreases as \(n\) increases to \(50\).
    One possible explanation is for \(n \leq 30\), the time span is too short for a periodic pattern to emerge in the data.
    However, for \(n = 20\), the model has more degrees of freedom than for \(n = 30\), leading the estimator to overfit the data and produce a good fit.
    From \(n = 40\) onwards, the time span is large enough for periodicity to emerge in the data, resulting in a better fit.
    We hypothesize that the fit improves further with an increasing time span.
    However, the computational costs of Algorithm~1 prevent a further analysis of the trend.    

    \subsection{Evaluation of research question Q2 and Q3}
    To analyze Q2 and Q3, we applied scheduling algorithms to both the actual dataset and the periodic dataset, where the latter consists of solutions obtained by applying Algorithm~1 to the original dataset. The results are summarized in Table~\ref{tab:resultsAlg2}. Grouped again by parameters \(k\) and \(n\), the table provides the average waiting time per vessel (measured in minutes) for each implemented setting. Specifically, we applied Algorithm~2 to the periodic dataset (column \emph{periodicOpt}), \emph{alternating} to the actual dataset (column \emph{alternating}), \emph{FIFO} to the actual dataset (column \emph{FIFO}), \emph{advFIFO} to the actual dataset (column \emph{advFIFO}), and Algorithm~ to the actual dataset (column \emph{realisedPeriodic}).

    Note that \emph{periodicOpt} is the average deviation in the most optimistic scenario in which the vessel arrival are truely periodic. 
    Thus, these values serve as benchmarks for comparison. 
    Regarding Q2, when examining the ratio of \emph{realisedPeriodic} to \emph{periodicOpt}, the highest observed ratio is \(3.39\), the lowest is \(1.34\), and the average is \(2.05\). Notably, the ratios observed for cases \(n = 20\) vessels are substantially higher; excluding these cases reduces the average ratio to \(1.67\).

    Regarding Q2, when examining the alternative policies, we observe that \emph{alternating} clearly outperforms all other implemented policies. 
    Comparing \emph{alternating} to \emph{periodicOpt}, the ratio ranges from \(2.12\) for \(n = 20\) to \(1.65\) for \(n = 50\).
    Furthermore, the \emph{advFIFO} policy consistently outperforms the \emph{FIFO} policy in most cases.
    The results indicate that periodic schedules perform relatively well on the actual dataset when compared to the optimal schedules derived from truly periodic arrivals.
     
\begin{table}[t]
\centering
\renewcommand{\arraystretch}{1.5} 
\setlength{\tabcolsep}{15pt} 
\begin{tabular}{l c c c}
\hline
k & n & Runtime & Fit \\ \hline
\multirow{4}{*}{2} 
  & 20 & 0.01 & 113.90 \\ 
  & 30 & 0.04 & 199.76 \\ 
  & 40 & 0.13 & 152.66 \\ 
  & 50 & 0.27 & 131.16 \\ \hline
\multirow{4}{*}{3} 
  & 20 & 0.25 & 99.66 \\ 
  & 30 & 2.61 & 187.55 \\ 
  & 40 & 12.45 & 144.27 \\ 
  & 50 & 54.11 & 124.59 \\ \hline
\multirow{4}{*}{4} 
  & 20 & 2.62 & 85.12 \\ 
  & 30 & 79.69 & 176.11 \\ 
  & 40 & 739.45 & 136.47 \\ 
  & 50 & 2990.27 & 118.33 \\ \hline
\end{tabular}
\caption{The results of the periodic matching problem, solved using Algorithm~1, are presented. Runtime is measured in seconds, while the fit is measured in minutes and represents the average deviation between the matching point and arrival time in the data.}
\label{tab:resultsAlg1}
\end{table}

\begin{sidewaystable}[h!]
\centering
\renewcommand{\arraystretch}{1.5} 
\setlength{\tabcolsep}{10pt} 
\begin{tabular}{l c c c c c c}
\hline
k & n & periodicOpt & alternating & FIFO & advFIFO & realisedPeriodic \\ \hline
\multirow{4}{*}{2} 
  & 20 & 3.09 & 6.55 & 7.39 & 7.59 & 10.47 \\ 
  & 30 & 6.18 & 8.72 & 10.01 & 9.29 & 11.36 \\ 
  & 40 & 9.40 & 11.83 & 13.50 & 12.33 & 12.64 \\ 
  & 50 & 9.01 & 14.84 & 16.59 & 15.76 & 16.57 \\ \hline
\multirow{4}{*}{3} 
  & 20 & 3.15 & 6.55 & 7.39 & 7.59 & 9.67 \\ 
  & 30 & 6.26 & 8.72 & 10.01 & 9.29 & 11.61 \\ 
  & 40 & 9.38 & 11.83 & 13.50 & 12.33 & 12.91 \\ 
  & 50 & 9.90 & 14.84 & 16.59 & 15.76 & 16.82 \\ \hline
\multirow{4}{*}{4} 
  & 20 & 3.63 & 6.55 & 7.39 & 7.59 & 11.19 \\ 
  & 30 & 6.15 & 8.72 & 10.01 & 9.29 & 11.83 \\ 
  & 40 & 8.74 & 11.83 & 13.50 & 12.33 & 13.03 \\ 
  & 50 & 9.88 & 14.84 & 16.59 & 15.76 & 16.23 \\ \hline
\end{tabular}
\caption{The table presents the results of the scheduling problem solved using various algorithms. The values represent the waiting time per vessel, measured in minutes. \emph{periodicOpt} refers to the results of Algorithm~2 applied to the periodic arrival data. \emph{alternating} represents the outcomes of the best alternating strategy. \emph{FIFO} denotes the first-in, first-out approach, while \emph{advanced FIFO} accounts for the next time period in its scheduling. \emph{realisedPeriod} represents the schedule derived from applying \emph{periodicOpt} to the original arrival data.}
\label{tab:resultsAlg2}
\end{sidewaystable}

    \section{Conclusions}\label{sec:conclusions}
    In this paper, we explored a comprehensive framework for periodic lock scheduling. 
    We proposed algorithms to identify regular streams that closely align with irregular arrival data. 
    Additionally, we developed algorithms that utilize this input data to compute optimal schedules for the lock. 
    We analyzed these algorithms both theoretically and through numerical experiments.
    We close the paper by mentioning a few open problems.

   From a practical perspective, the simplifications introduced in our model limit its direct applicability for practitioners.
   Specifically, we need to integrate the lock’s capacity along with individual vessel sizes and their respective loading times.
   Furthermore, the integration between periodic lock scheduling and unforeseen delays during vessel entry into the lock chamber—addressed by Buchem et al. \cite{buchem2022vessel}—represents another direction for future research.

    From a theoretical perspective, the complexity of both the periodic scheduling problem and the periodic matching problem remains open and is of interest.
    Furthermore, all the practical extensions mentioned above raise new questions regarding algorithm design and complexity bounds.
    
\bibliographystyle{abbrv}
\bibliography{cyclic.bib}
\end{document}